\documentclass[a4paper,UKenglish,11pt]{article}

\usepackage{microtype}

\usepackage{fullpage}

\usepackage{graphicx}
\usepackage{enumitem}
\usepackage{amsmath} 
\usepackage{amssymb}
\usepackage{amsthm}
\usepackage{amsfonts}
\usepackage{bbm,graphicx,bm }
\usepackage{todonotes}

\usepackage{xcolor}
\definecolor{dark-red}{rgb}{0.4,0.15,0.15}
\definecolor{dark-blue}{rgb}{0.15,0.15,0.4}
\definecolor{medium-blue}{rgb}{0,0,0.5}
\definecolor{gray}{rgb}{0.5,0.5,0.5}

\newcommand{\apolytopematrices}{PM}
\newcommand{\R}{\mathbb{R}}

\newcommand{\completegraph}{K}
\newcommand{\discretegraph}{D}
\newcommand{\parsetree}{t}
\newcommand{\leafposition}{p}

\newcommand{\neighborhood}{N}
\newcommand{\closedneighborhood}{\overline{N}}
\newcommand{\avertex}{v}
\newcommand{\anothervertex}{u}
\newcommand{\vertexset}{V}
\newcommand{\cosetrepresentatives}{T}

\newcommand{\subtreelikeset}{M}
\newcommand{\eyield}{\yield_{\enrichmentvtx}}
\newcommand{\treewidth}{k}
\newcommand{\pathwidth}{k}
\newcommand{\decwidth}{k}
\newcommand{\cLang}{\pi}
\newcommand\sym{\symmetricgroup}
\newcommand{\annotatedbags}{\mathfrak{B}}

\newcommand{\anotherpolytope}{Q}
\newcommand{\setvectors}{W}

\newcommand{\setinequalities}{\mathcal{E}}
\newcommand{\variables}{\mathcal{B}}
\newcommand{\grammarrelation}{R}
\newcommand{\initialvariable}{\avariable_1}
\newcommand{\avariable}{B}
\newcommand{\allgrammars}{\mathbb{G}}

\newcommand{\allregulargrammars}{\mathbb{R}\mathbb{G}}
\newcommand{\stringfy}[1]{\mathit{str}(#1)}
\newcommand{\indexgroup}{\mathcal{I}}

\makeatletter
\newcommand\footnoteref[1]{\protected@xdef\@thefnmark{\ref{#1}}\@footnotemark}
\makeatother

\newcommand{\realvariables}{\mathcal{X}}
\newcommand{\otherrealvariables}{\mathcal{Y}}

\newcommand{\avector}{v}

\newcommand{\treelikeset}{U}

\newcommand{\anotherpermutation}{\beta}

\newcommand{\agroup}{G}
\newcommand{\anothergroup}{H}
\newcommand{\subgroup}{\sqsubseteq}

\newcommand{\asubset}{S}
\newcommand{\defineequal}{\stackrel{\mathit{def}}{=}}

\newcommand{\astring}{w}

\newcommand{\asymbol}{a}

\newcommand{\polytope}{P}

\newcommand{\aposition}{p}

\newcommand{\agraph}{X}
\newcommand{\anothergraph}{Y}

\newcommand{\arity}{\mathfrak{a}}
\newcommand{\maximumdegree}{\Delta}

\newcommand{\stringfication}{\mathrm{str}}

\newcommand{\convexhull}{\mathrm{conv}}
\newcommand{\extensioncomplexity}{\mathrm{xc}}
\newcommand{\embeddingcomplexity}{\mathrm{gec}}

\newcommand{\regsgc}{\mathrm{reg\mbox{-}sgc}}
\newcommand{\sgc}{\mathrm{sgc}}
\newcommand{\contextfreegrammar}{\mathfrak{G}}
\newcommand\cfg{\contextfreegrammar}

\newcommand{\yield}{\mathrm{yield}}
\newcommand{\grammar}{\mathfrak{G}}

\newcommand{\apolytope}{P}

\newtheorem{theorem}{Theorem}
\newtheorem{definition}[theorem]{Definition}
\newtheorem{lemma}[theorem]{Lemma}
\newtheorem{corollary}[theorem]{Corollary}
\newtheorem{observation}[theorem]{Observation}
\newtheorem{proposition}[theorem]{Proposition}
\newtheorem{claim}[theorem]{Claim}
\newtheorem{problem}[theorem]{Problem}
\newtheorem{myremark}[theorem]{Remark}
\newtheoremstyle{named}{}{}{\itshape}{}{\bfseries}{.}{.5em}{Restatement of #1 \thmnote{#3 }}
\theoremstyle{named}
\newtheorem*{retheorem}{Theorem}

\newtheorem*{reremark}{Remark}

\providecommand{\keywords}[1]{\textbf{\textrm{Keywords.}} #1}

\newcommand{\mylang}{{\mathcal{L}}} 

\newcommand{\Nplus}{{\mathbb{N}_+}}

\newcommand{\treedecomposition}{\mathbf{t}}
\newcommand{\td}{\treedecomposition}
\newcommand{\annotatedtreedecomposition}{\hat{\treedecomposition}}
\newcommand{\etd}{\annotatedtreedecomposition}
\newcommand{\reals}{\mathbb{R}}

\newcommand{\leaves}{\mathit{leaves}}

\newcommand{\emptysymbol}{\varepsilon}

\newcommand{\alphabet}{\Sigma}

\newcommand{\composedUprime}{%
  \mathrel{\vbox{\offinterlineskip\ialign{%
    \hfil##\hfil\cr
    $\scriptscriptstyle\circ$\cr
    \noalign{\kern0.1ex}
    $\boldU'$\cr
}}}}

\newcommand{\composedU}{%
  \mathrel{\vbox{\offinterlineskip\ialign{%
    \hfil##\hfil\cr
    $\scriptscriptstyle\circ$\cr
    \noalign{\kern0.1ex}
    $\boldU$\cr
}}}}

\newcommand{\composedW}{%
  \mathrel{\vbox{\offinterlineskip\ialign{%
    \hfil##\hfil\cr
    $\scriptscriptstyle\circ$\cr
    \noalign{\kern0.1ex}
    $\boldW$\cr
}}}}

\newcommand{\composedD}{%
  \mathrel{\vbox{\offinterlineskip\ialign{%
    \hfil##\hfil\cr
    $\scriptscriptstyle\circ$\cr
    \noalign{\kern0.1ex}
    $\boldD$\cr
}}}}

\usepackage{calligra}
\DeclareMathAlphabet{\mathcalligra}{T1}{calligra}{m}{n}
\newcommand*{\largerdot}{\raisebox{-0.25ex}{\scalebox{1.5}{$\cdot$}}}

\newcommand{\morphismvertices}{%
  \mathrel{\vbox{\offinterlineskip\ialign{%
    \hfil##\hfil\cr
    $\scriptscriptstyle\largerdot$\cr
    \noalign{\kern-0.3ex}
    $\mu$\cr
}}}}

\newcommand{\edgeset}{E}

\newcommand{\boldD}{\mathbf{D}}

\newcommand{\boldW}{\mathbold{W}}

\newcommand{\boldU}{{\mathbf{U}}}

\newcommand{\lang}{L} 

\usepackage{relsize}

\newcommand{\bN}{\mathbb{N}}
\newcommand{\cO}{\mathcal{O}}

\newcommand*{\defeq}{\mathrel{\vcenter{\baselineskip0.5ex \lineskiplimit0pt
                     \hbox{\scriptsize.}\hbox{\scriptsize.}}}%
                     =}

\newenvironment{claimproof}{\begin{proof}\renewcommand{\qedsymbol}{\claimqed}}{\end{proof}\renewcommand{\qedsymbol}{\plainqed}}

\let\plainqed\qedsymbol

\newcommand{\card}[1]{|#1|}

\newcommand{\permutation}{\alpha}

\newcommand{\myterms}{\mathit{Ter}}
\newcommand{\positions}{\mathit{Pos}}

\newcommand{\emptystring}{{\lambda}}

\newcommand{\symmetricgroup}{\mathbb{S}}
\newcommand{\isomorphisms}{\mathop{\mathrm{Iso}}}
\newcommand{\automorphisms}{\mathop{\mathrm{Aut}}}

\newcommand{\aterm}{t}
\newcommand{\annotatedterm}{\hat{\aterm}}
\newcommand{\eterm}{\annotatedterm}

\newcommand{\productionrules}{\mathcal{R}}

\newcommand{\alternatinggroup}{\mathbb{A}}

\newcommand{\poormap}{\rho}

\newcommand{\graphclass}{\mathcal{X}}
\newcommand{\powerset}{\mathcal{P}}

\newcommand{\stab}{\mathrm{stab}}

\newcommand{\apermutation}{\alpha}
\newcommand\aperm{\apermutation}
\newcommand{\permutedcoordinates}{\mathrm{Perm}}

\newcommand\enrichmentvtx{\nu}
\newcommand\evtx{\enrichmentvtx}

\newcommand\enrichmentautomorphism{\mu}
\newcommand\eaut{\enrichmentautomorphism}

\newcommand{\blang}{L}

\title{Compressing Permutation Groups into Grammars and Polytopes. \\ A Graph Embedding Approach%
\thanks{Lars Jaffke acknowledges support from the Bergen Research Foundation. Mateus de Oliveira Oliveira
acknowledges support from the Bergen Research Foundation and from the Research Council of Norway (Project Number 288761). 
Hans Raj Tiwary is partially supported by the grant GA\v{C}R 17-09142S.}}

\author{
Lars Jaffke$^{1}$ \hspace{0.5cm} Mateus de Oliveira Oliveira$^{1}$ \hspace{0.5cm} Hans Raj Tiwary$^{2}$ \\
\\ 
$^1$Department of Informatics, University of Bergen, Bergen, Norway \\ \{lars.jaffke,mateus.oliveira\}@uib.no \\  
$^2$Department of Applied Mathematics, Charles University, Prague, Czech Republic \\ {hansraj@kam.mff.cuni.cz}\\
}

\begin{document}

\maketitle

\begin{abstract} 
It can be shown that each permutation group $\agroup \subgroup \symmetricgroup_n$
can be embedded, in a well defined sense, in a {\em connected} graph
with $O(n+|G|)$ vertices. Some groups, however, require much fewer vertices.
For instance, $\symmetricgroup_n$ itself can be embedded in the $n$-clique $K_n$,
a connected graph with $n$ vertices.

In this work, we show that the minimum size of a context-free grammar generating a finite permutation 
group $\agroup\subgroup \symmetricgroup_n$ can be upper bounded by three structural parameters of {\em connected} graphs 
embedding $\agroup$: the number of vertices, the treewidth, and the maximum degree.
More precisely, we show that any permutation group $\agroup \subgroup \symmetricgroup_n$ that can be embedded 
into a connected graph with $m$ vertices, treewidth $\treewidth$, and maximum degree $\maximumdegree$,
can also be generated by a context-free grammar of size $2^{O(\treewidth\maximumdegree\log\maximumdegree)}\cdot m^{O(\treewidth)}$.
By combining our upper bound with a connection established by Pesant, Quimper, Rousseau and Sellmann \cite{Pesant2009} 
between the extension complexity of a permutation group and the grammar complexity of a formal language, we also 
get that these permutation groups can be represented by polytopes of extension complexity 
$2^{O(\treewidth\maximumdegree\log\maximumdegree)}\cdot m^{O(\treewidth)}$. 

The above upper bounds can be used to provide trade-offs between the index of permutation groups, and the number of 
vertices, treewidth and maximum degree of connected graphs embedding these groups. 
In particular, by combining our main result with a celebrated $2^{\Omega(n)}$ lower bound on the grammar 
complexity of the symmetric group $\symmetricgroup_n$ due to Glaister and Shallit \cite{GlaisterShallit1996}
we have that connected graphs of treewidth $o(n/\log n)$ and maximum degree $o(n/\log n)$ embedding subgroups 
of $\symmetricgroup_n$ of index $2^{cn}$ for some small constant $c$ must have $n^{\omega(1)}$
vertices. This lower bound can be improved to exponential on graphs of treewidth $n^{\varepsilon}$ for $\varepsilon<1$ and 
maximum degree $o(n/\log n)$. 
\\
\\
 %% \keywords are mandatory in final camera-ready submission
\end{abstract}

\newpage
\section{Introduction}
\label{section:Introduction}

Let $\symmetricgroup_n$ be the set of permutations of the set $\{1,...,n\}$
and $\stringfication(\symmetricgroup_n)$ be the set of strings in $\{1,...,n\}^n$ encoding
permutations in $\symmetricgroup_n$. The search for minimum size grammars generating the language 
$\stringfication(\symmetricgroup_n)$ has sparked a lot of interest in the automata theory and in the 
complexity theory communities, both in the study of lower bounds \cite{Ellul2004,LovettShallit2019,Filmus2011lower},
and in the study of upper bounds \cite{gruber2018minimal,asveld2008generating,asveld2006generating}. 
In particular, a celebrated result due to Ellul, Krawetz and Shallit \cite{Ellul2004} states that any 
context-free grammar generating the language $\stringfication(\symmetricgroup_n)$ must have size $2^{\Omega(n)}$.
In this work we complement this line of research by establishing upper bounds for the size of context-free grammars 
representing a given subgroup $\agroup \subgroup \symmetricgroup_n$. These upper bounds are stated 
in terms of three structural parameters of connected graphs embedding $\agroup$: the number 
of vertices, the treewidth and the maximum degree.

We say that a permutation group $\agroup\subgroup \symmetricgroup_n$ can be embedded in
a graph $\agraph$ with vertex set $[m]=\{1,...,m\}$, if $m\geq n$ 
and $\agroup$ is equal to the restriction of the automorphism group of $\agraph$ 
to its first $n$ vertices $[n]=\{1,...,n\}$.
A more precise definition of the notion of graph embedding is given in Section \ref{section:GraphEmbeddable}. 
For a given class of {\em connected} graphs $\graphclass$, the $\graphclass$-embedding complexity of $\agroup$, 
denoted by $\embeddingcomplexity_{\graphclass}(\agroup)$, is defined as the minimum $m$ such that $\agroup$ 
can be embedded in an $m$-vertex graph $\agraph\in\graphclass$. 

Given an alphabet $\alphabet$, the {\em symmetric grammar complexity} (SGC) of a formal language $\lang\subseteq \alphabet^n$ measures 
the minimum size of a context-free grammar accepting a permuted version of $\lang$. As a matter of comparison, we note that online Turing machines working in space $s$ and with access to a 
stack have symmetric grammar complexity $2^{O(s)}$ \cite{GoldwurmPalanoSantini2001circuit}. 
In this setting, the machine reads the input string $\astring\in \alphabet^n$ from left to right, one symbol at a time. 
While reading this string, symbols can be pushed into or popped from the stack. The transitions relation depends 
on the current state, on the symbol being read at the input, and on the symbol being read at the top of the stack.
The caveat is that the number of symbols used in the stack (which can be up to $n$) is not counted in the space bound $s$, which can be 
much smaller than $n$ (say $s=O(\log n)$). 
The SGC of a language $\lang\subseteq \alphabet^{n}$ is also polynomially related to the minimum size of a read-once branching program with a stack
accepting $\lang$ (see for instance \cite{Mengel2013arithmetic}).

\subsection{Our Results}

We show that the automorphism group of any graph with $n$ vertices, maximum degree $\maximumdegree$ and treewidth $\treewidth$
has symmetric grammar complexity at most $2^{O(\treewidth\maximumdegree\log\maximumdegree)}\cdot n^{O(\treewidth)}$ 
(Theorem \ref{theorem:AutomorphismsTreewidth}). More generally, we show that the SGC of groups that can be embedded in
$m$-vertex graphs of maximum degree $\maximumdegree$ and treewidth $\treewidth$ is at most 
$2^{O(\treewidth\maximumdegree\log\maximumdegree)}\cdot m^{O(\treewidth)}$ (Theorem \ref{theorem:MainTheoremEmbeddableGroup}). 

In linear programming theory, it can be shown that there are interesting polytopes 
$\apolytope \subseteq \R^n$, which can only be defined with an exponential (in $n$) number 
of inequalities, but which can be cast as a linear projection of a higher dimensional 
polytope $\anotherpolytope$ that can be defined with polynomially many variables and constraints. 
Such a polytope $\anotherpolytope$ is called an extended formulation of $\apolytope$.
Extended formulations of polynomial size play a crucial role in combinatorial optimization
because they provide an unified framework to obtain polynomial time algorithms for a large 
variety of combinatorial problems. For this reason, extended formulations of polytopes associated
with formal languages and with groups have been studied intensively during the past decades, 
both from the perspective of lower bounds \cite{Rothvoss2013,FioriniMassarPokuttaTiwaryWolf2015,Yannakakis1991expressing,Pokutta2013note,Avis2015extension,
Braun2015approximation,KaibelWeltge2015short}, and from the perspective of upper bounds
\cite{Barahona1993cuts,DezaLaurent1997geometry,BalasPulleyblank1983perfectly,Barahona1993cuts,Pulleyblank1993formulations,Faenza2009extended,Yannakakis1991expressing,Cheung2003subtour,Conforti2009network}. 

By combining our main theorem \ref{theorem:MainTheoremEmbeddableGroup} with a connection established by
Pesant, Quimper, Rousseau and Sellmann \cite{Pesant2009}  between the extension complexity of a 
permutation group and the grammar complexity of a formal language, we show that any permutation group
that can be embedded in a connected graph with $m$ vertices, treewidth $\treewidth$, 
and maximum degree $\maximumdegree$ can be represented by polytopes of extension complexity 
$2^{O(\treewidth\maximumdegree\log\maximumdegree)}\cdot m^{O(\treewidth)}$ (Theorem \ref{theorem:EmbeddabilityVsPolytopes}). 

By combining our upper bound from Theorem \ref{theorem:MainTheoremEmbeddableGroup} with the $2^{\Omega(n)}$ lower bound from \cite{Ellul2004},
we obtain an interesting complexity theoretic trade-off relating the index of a permutation group with the size, treewidth and maximum degree 
of a graph embedding this group (Theorem \ref{theorem:TradeoffTreewidth}). As a corollary of this trade-off, we show that subgroups 
of $\symmetricgroup_n$ with index up to $2^{cn}$ for some small constant $c$ have superpolynomial graph embedding complexity on classes 
of graphs with treewidth $o(n/\log n)$ and maximum degree $o(n/\log n)$ (Corollary \ref{corollary:LowerBoundOne}). 
Additionally, this lower bound can be improved from super-polynomial to exponential on classes of graphs of 
treewidth $n^{\varepsilon}$ (for $\varepsilon< 1$) and maximum degree $o(n/\log n)$ (Corollary \ref{corollary:LowerBoundTwo}). 
In particular, Corollary \ref{corollary:LowerBoundTwo} implies exponential lower bounds for minor-closed families of connected graphs (which have treewidth $\sqrt{n}$).

\subsection{Related Work}

Proving lower bounds for the size of graphs embedding a given permutation group is a challenging
and still not well understood endeavour. It is worth noting that it is still not known whether 
the alternating group $\alternatinggroup_n$ can be embedded in a graph with $n^{O(1)}$ vertices.
We note that by solving an open problem stated by Babai in \cite{Babai1981abstract}, Liebeck has shown that 
any graph whose automorphism group is isomorphic to the alternating group (as an abstract group) must have
at least $2^{\Omega(n)}$ vertices \cite{Liebeck1983graphs}. Nevertheless, a similar result has not yet been 
obtained in the setting of graph embedding of groups, and indeed, constructing an explicit sequence of groups 
that have superpolynomial graph embedding complexity is a long-standing open problem \cite{Babai1995automorphism}.
Our results in Corollary \ref{corollary:LowerBoundOne} and Corollary \ref{corollary:LowerBoundTwo} 
provide unconditional lower bounds for interesting classes of graphs for any group 
of relatively small index (index at most $2^{cn}$ for some small enough constant $c$).

The crucial difference between the abstract isomorphism setting considered in \cite{Liebeck1983graphs}
and our setting is in the way in which graphs are used to represent groups. 
In the setting of \cite{Liebeck1983graphs}, given a group $\agroup$, the goal is to construct a graph $\agraph$ whose 
automorphism group is {\em isomorphic} to $G$. On the other hand, in the graph embedding setting, we want the group $G$ to 
be {\em equal} to the action of the automorphism group $\automorphisms(\agraph)$ on its first $[n]$ vertices.
In the abstract isomorphism setting it has been shown by Babai that any class of graphs $\graphclass$ excluding a fixed graph $H$ as a 
minor, there exists some finite group which is not isomorphic to the automorphism group of any graph in $\graphclass$ 
\cite{Babai1974automorphism}. Our Corollary \ref{corollary:LowerBoundTwo} can be regarded as a result in this spirit 
in the context of graph embedding. While the lower bound stated in Corollary \ref{corollary:LowerBoundTwo} also applies to
graphs that are not minor closed, this lower boud is only meaningful for graphs of maximum degree  at most $o(n/\log n)$. 

We observe that in Theorem \ref{theorem:MainTheoremEmbeddableGroup} an exponential dependence on the maximum degree parameter 
$\maximumdegree$ is unavoidable. Indeed, as stated above, the symmetric grammar complexity of the language $\stringfication(\symmetricgroup_n)$ is $2^{\Theta(n)}$.
On the other hand, for each $n\in \Nplus$, the symmetric group $\symmetricgroup_n$ can be embedded in the star 
graph $\completegraph_{n,1}$ with vertex set $\vertexset(\completegraph_{n,1}) = \{1,...,n+1\}$, 
and edge set $\edgeset(\completegraph_{n,1}) = \{\{i,n+1\}\;:\; i\in \{1,\dots,n\}\}$, which is a connected graph of treewidth $1$.
Nevertheless, it is not clear to us whether the logarithmic factor $\log \maximumdegree$ can be shaved from the exponent of the upper bound 
$2^{O(\treewidth\maximumdegree\log\maximumdegree)}\cdot m^{O(\treewidth)}$. 
We also note that the connectedness requirement is also crucial for our upper bounds since $\symmetricgroup_n$ can be embedded in the discrete
graph $\discretegraph_n$ with vertex set $\discretegraph_n = \{1,...,n\}$, and edge set $\edgeset(\discretegraph_n) = \emptyset$.

\section{Preliminaries}

We let $\bN$ denote the set of non-negative integers 
and $\Nplus = \bN \setminus \{0\}$ denote the set of positive integers. 
For each $n\in \Nplus$, we let $[n] = \{1,...,n\}$. For each finite set $S$ we let
$\powerset(S) = \{ S'\;:\; S'\subseteq S\}$ denote the set of all subsets of
$S$.  For each set $S$ and each $k\in \bN$, we let $\binom{S}{k}  =
\{S'\subseteq S\;:\; |S'|=k\}$ be the set of subsets of $S$ of size $\treewidth$ and 
$\binom{S}{\le k} = \bigcup_{i = 0}^k \binom{S}{i}$ the set of subsets of size at most $\treewidth$.
For a function $f\colon X \to Y$ and a set $X' \subseteq X$, we denote by $f|_{X'}$ 
the {\em restriction of $f$ to $X'$}, i.e. the function $f|_{X'} \colon X' \to Y$ with 
$f|_{X'}(x) = f(x)$ for each $x \in X'$. 

\medskip
\noindent{\bf Prefix Closed Sets.}
For each $r\in \Nplus$, we let $[r]^*$ be the set of all strings over $[r]$, including the empty 
string $\emptystring$. Let $p$ and $u$ be strings in $[r]^*$. We say that $p$ is a {\em prefix} 
of $u$ if there exists  $q\in [r]^*$ such that $u=pq$. Note that $u$ is a prefix of itself, and that 
the empty string $\emptystring$ is a prefix of each string in $[r]^*$. 
A non-empty subset $\treelikeset \subseteq [r]^*$ is {\em prefix closed} if 
for each $u\in \treelikeset$, each prefix of $u$ is also in $\treelikeset$. 
We note that the empty string $\emptystring$ is an element of any prefix 
closed subset of $[r]^*$.  We say that $\treelikeset\subseteq [r]^*$ is {\em well numbered} 
if for each $p\in [r]^*$ and each $j\in [r]$, the
presence of $pj$ in $\treelikeset$ implies that $p1,...,p(j-1)$ also belong to $\treelikeset$.

\medskip
\noindent{\bf Tree-Like Sets.}
We say that a subset $\treelikeset\subseteq [r]^*$ is {\em tree-like} if $\treelikeset$ is both 
prefix-closed and well-numbered. Let $\treelikeset$ be a tree-like subset of $[r]^*$. 
If $pj\in \treelikeset$, then we say that $pj$ is a {\em child} of $p$, or interchangeably, 
that $p$ is the {\em parent} of $pj$.  If
$pu\in \treelikeset$ for $u\in [r]^*$, then we say that $pu$ is a descendant of
$p$. For a node $p\in \treelikeset$ we let $\treelikeset|_p =\{pu\in \treelikeset\; :\; u\in [r]^*\}$
denote the set of all descendants of $p$. 
Note that $p$ is a descendant of itself and therefore, $p\in \treelikeset|_p$. 
A {\em leaf} of $\treelikeset$ is a node $p\in \treelikeset$ without children. We let $\leaves(\treelikeset)$ be the 
set of leaves of $\treelikeset$,
and $\leaves(\treelikeset,p)$ be the set of leaves which are descendants of $p$.

\medskip
\noindent{\bf Terms.} 
Let $\alphabet$ be a finite set of symbols. An {\em $r$-ary term} over $\alphabet$ 
is a function $\aterm:\positions(\aterm)\rightarrow \alphabet$ whose domain $\positions(\aterm)$ 
is a tree-like subset of $[r]^*$. 
We denote by $\myterms(\alphabet)$ the set of all terms 
over $\alphabet$. If $\aterm_1,...,\aterm_r$ are terms in $\myterms(\alphabet)$, and $\asymbol\in \alphabet$,
then we let $\aterm  = \asymbol(\aterm_1,...,\aterm_r)$ be the term in $\myterms(\alphabet)$ which 
is defined by setting $\aterm(\emptystring) = \asymbol$ and $\aterm(j\aposition) = \aterm_j(\aposition)$ 
for each $j\in [r]$ and each $\aposition\in \positions(t_j)$.

\section{Embedding Permutation Groups in Graphs}
\label{section:GraphEmbeddable}

For each finite set $\Gamma$, we let $\symmetricgroup(\Gamma)$ be the group of permutations of $\Gamma$. 
If $\Omega\subseteq \Gamma$  and $\apermutation\in \symmetricgroup(\Gamma)$, then we say that 
$\apermutation$ stabilizes $\Omega$ setwise if $\apermutation(\Omega) = \Omega$. Alternatively, we say 
that $\Omega$ is invariant under $\alpha$. We let $\alpha_{\Omega}$ be the permutation in $\symmetricgroup(\Omega)$ 
which is defined by setting $\alpha_{\Omega}(i)=\alpha(i)$ for each $i\in \Omega$. In other words,
$\alpha_{\Omega}$ is the restriction of $\alpha$ to $\Omega$.  
If $\agroup$ is a subgroup of $\symmetricgroup(\Gamma)$, then we let $\stab(\agroup,\Omega)$ be the set 
of permutations in $\agroup$ that stabilize $\Omega$ setwise. We say that a group $\agroup$ stabilizes 
$\Omega$ if $\stab(\agroup,\Omega) = \agroup$. Alternatively, we say that $\Omega$ is invariant under $\agroup$. 
We let $\agroup|_{\Omega} = \{\alpha|_{\Omega} \;:\; \alpha \in \agroup\}$ be the set of restrictions of permutations in 
$\agroup$ to $\Omega$. 
In what follows, for each $n\in \Nplus$ we write $\symmetricgroup_n$ to denote $\symmetricgroup([n])$.
\\ 

\noindent{\bf Graphs.} 
Let $m\in \Nplus$. An $m$-vertex graph is a pair $\agraph = ([m], \edgeset(\agraph))$,
where $\edgeset(\agraph) \subseteq \binom{[m]}{2}$.
\\

\noindent{\bf Isomorphisms and Automorphisms.}
If $\agraph$ and $\anothergraph$ are two $m$-vertex graphs, then an {\em isomorphism} between $\agraph$
and $\anothergraph$ is a permutation $\alpha\in \symmetricgroup_m$ such that for each
$\{i,j\}\in \binom{[m]}{2}$, $\{i,j\}\in \agraph$ if and only if
$\{\alpha(i),\alpha(j)\}\in \anothergraph$.  An {\em automorphism} of $\agraph$ is an isomorphism
between $\agraph$ and $\agraph$.  We let $\isomorphisms(\agraph,\anothergraph)$ denote the set of all
isomorphisms between $\agraph$ and $\anothergraph$, and let $\automorphisms(\agraph) =
\isomorphisms(\agraph,\agraph)$ be the set of automorphisms of $\agraph$. If $\Omega\subseteq [m]$ 
is invariant under $\automorphisms(\agraph)$ then we define
$\automorphisms(\agraph,\Omega) = \automorphisms(\agraph)|_{\Omega} = \{\alpha|_{\Omega}\;:\; \alpha\in \automorphisms(\agraph)\}$.

\begin{definition}
\label{definition:GraphEmbeddable}
Let $\agroup$ be a subgroup of $\symmetricgroup_n$ and $\agraph$ be a connected $m$-vertex graph where $m\geq n$. 
We say that $\agroup$ is embeddable in $\agraph$ if $\automorphisms(\agraph,[n]) = \agroup$. 
\end{definition}

In other words, $\agroup$ is embeddable in $\agraph$ if the image of action of the automorphism group of $\agraph$ on 
its first $n$ vertices is equal to $\agroup$. 
We note that the requirement that the graph $\agraph$ of Definition \ref{definition:GraphEmbeddable} is connected is 
crucial for our applications.

Let $\graphclass$ be a class of connected graphs and $\agroup$ be a subgroup of $\symmetricgroup_n$. 
We say that $\agroup$ is {\em $\graphclass$-embeddable} if
there exists some graph $\agraph\in \graphclass$ such that $\agroup$ is embeddable in $\agraph$. 
The {\em $\graphclass$ embedding complexity} of $\agroup$, denoted by $\embeddingcomplexity_{\graphclass}(\agroup)$ 
is the minimum $m$ such that $\agroup$ is embeddable in a graph $X\in \graphclass$ with at most $m$ vertices.
If no such a graph $X\in \graphclass$  exists, then we set $\embeddingcomplexity_{\graphclass}(\agroup)=\infty$.

\section{Using Grammars to Represent Finite Permutation Groups}
\label{section:GrammarGroups}

A {\em context-free grammar} is a $4$-tuple $\contextfreegrammar = (\alphabet,\variables,\grammarrelation,\initialvariable)$
where $\alphabet$ is a finite set of symbols, $\variables$ is a finite set of variables, 
$\grammarrelation \subseteq \variables \times (\alphabet \cup \variables)^*$ is a finite set of {\em production rules}, 
and $\initialvariable\in \variables$  is the {\em initial variable} of $\contextfreegrammar$.
The notion of a string $w$ generated by $\contextfreegrammar$ can be defined with basis on the notions 
of {\em $\contextfreegrammar$-parse-tree} and {\em yield} of a $\contextfreegrammar$-parse-tree, which are
inductively defined as follows.

\begin{enumerate}
\item For each $\asymbol\in \alphabet$ the term $\aterm:\{\emptystring\}\rightarrow \alphabet$ which sets 
	$\aterm(\emptystring) = \asymbol$ is a $\contextfreegrammar$-parse-tree. Additionally, $\yield(\aterm)  = \asymbol$.
\item If $\emptysymbol$ is the empty symbol then the term $\aterm:\{\emptystring\}\rightarrow \alphabet$ which 
	sets $\aterm(\emptystring) = \emptysymbol$ is a $\contextfreegrammar$-parse-tree. Additionally, $\yield(\aterm) = \emptysymbol$. 
\item If $\aterm_1,...,\aterm_r$ are $\contextfreegrammar$-parse-trees and 
	$\avariable\rightarrow \aterm_1(\emptystring)\aterm_2(\emptystring)...\aterm_r(\emptystring)$ is a production 
rule in $\grammarrelation$, then the term 
$\aterm = \avariable(\aterm_1,...,\aterm_r)$ is a $\contextfreegrammar$-parse-tree. 
Additionally, $$\yield(\aterm) = \yield(\aterm_1)\cdot \yield(\aterm_2)\cdot ... \cdot \yield(\aterm_r).$$
In other words, the yield of $\aterm$ is the concatenation of the yields of the subterms $\aterm_1,\dots,\aterm_r$. 
\end{enumerate}

We say that a $\contextfreegrammar$-parse-tree $\aterm$ is {\em accepting} if $\aterm(\emptystring) = \initialvariable$. 
We say that a string $w\in \alphabet^*$ is generated by $\contextfreegrammar$ if there is an accepting $\contextfreegrammar$-parse-tree
with $\yield(\aterm)  = w$. The language generated by $\contextfreegrammar$ is the set 
$\mylang(\contextfreegrammar) = \{w\in \alphabet^*\;:\; w\mbox{ is generated by $\contextfreegrammar$}\}$
of strings generated by $\contextfreegrammar$. The size of $\contextfreegrammar$ is defined as 
$$|\contextfreegrammar| = \sum_{(\avariable,u)\in \grammarrelation} (1+|u|) \log(|\alphabet|+|\variables|),$$ where
$|u|$ is the number of symbols/variables in $u$, $|\alphabet|$ is the number of elements in $\alphabet$ and $|\variables|$
is the number of elements in $\variables$. 
We denote by $\allgrammars(\alphabet)$ the set of context-free grammars over the alphabet $\alphabet$. \\

A context-free grammar $\contextfreegrammar$ is said to be regular if each production rule is either 
of the form $(\avariable,a)$ for some $\avariable\in \variables$ and $a\in \alphabet$, or of the form $(\avariable,a\avariable')$ 
for some $\avariable,\avariable'\in \variables$ and some $a\in \alphabet$. We denote by
$\allregulargrammars(\alphabet)$ the set of regular context-free grammars over the alphabet $\alphabet$. \\

\noindent{\bf Complexity Measures.}
If $\apermutation\in \symmetricgroup_n$ and $w\in \alphabet^{n}$ then 
we let $\permutedcoordinates(w,\apermutation) \defineequal w_{\apermutation(1)}w_{\apermutation(2)}...w_{\apermutation(n)}$
be the string obtained by permuting the positions of $w$ according to $\apermutation$. If $\blang\subseteq \alphabet^{n}$
then we let $\permutedcoordinates(\blang,\apermutation) \defineequal \{\permutedcoordinates(w,\apermutation)\;:\; w\in \blang\}$. 
In other words, $\permutedcoordinates(\blang,\apermutation)$ 
is the language obtained by permuting the positions of each string $w\in \blang$ according 
to $\apermutation$. The {\em symmetric grammar complexity} of a language $\blang\subseteq \alphabet^n$
is defined as the minimum size of a context-free grammar generating $\permutedcoordinates(\blang,\apermutation)$ 
for some  $\apermutation\in \symmetricgroup_n$. More precisely, 
$$\sgc(\blang) =\min \{|\contextfreegrammar|\;:\;\apermutation\in \symmetricgroup,\;\contextfreegrammar\in \allgrammars(\alphabet),\;
\mylang(\contextfreegrammar) = \permutedcoordinates(\blang,\apermutation)\}.$$

Analogously, the {\em symmetric regular grammar complexity} of a language $\blang\subseteq \alphabet^n$  is 
defined as the minimum size of a {\em regular} grammar generating $\permutedcoordinates(\blang,\apermutation)$ for 
some $\apermutation\in \symmetricgroup_n$.
$$\regsgc(\blang) = \min \{|\contextfreegrammar|:\apermutation\in \symmetricgroup,\;\contextfreegrammar\in \allregulargrammars(\alphabet),\;
\mylang(\contextfreegrammar) = \permutedcoordinates(\blang,\apermutation)\}.$$

We note that the {\em symmetric regular grammar complexity} of a language $\lang\subseteq \alphabet^n$ is polynomially
related to the minimum size of an acyclic non-deterministic finite automaton accepting some permuted version of $\lang$, or equivalently to the 
minimum size of a non-deterministic read-once oblivious branching program accepting $\lang$. On the other hand, the symmetric 
context-free complexity of a language $\lang$ is polynomially related to the minimum size of a pushdown automaton 
accepting some permuted version of $\lang$.

Let $\apermutation:[n]\rightarrow [n]$ be a permutation in $\symmetricgroup_n$. 
We let
$$\stringfy{\apermutation}= \apermutation(1)\apermutation(2)...\apermutation(n)\in [n]^{n}$$ 
be the string associated with $\apermutation$. For each group $\agroup\subgroup \symmetricgroup_n$ we let 
$\stringfy{\agroup} = \{\stringfy{\apermutation}\;:\; \apermutation\in \agroup\}$ be the language associated with 
$\agroup$. The symmetric grammar complexity of $\agroup$ is defined as 
$\sgc(\agroup) \defineequal \sgc(\stringfy{\agroup})$. 
Analogously, the regular 
grammar complexity of $\agroup$ is 
defined as $\regsgc(\agroup) \defineequal \regsgc(\stringfy{\agroup})$. 

If $\anotherpermutation:[n]\rightarrow [n]$ and $\gamma:[n]\rightarrow [n]$
are permutations in $\symmetricgroup_n$, then we let $\anotherpermutation \circ \gamma$ be the permutation that sends each 
$i\in [n]$ to the number $\anotherpermutation(\gamma(i))$. If $\asubset$ is a subset of $\symmetricgroup_n$, 
we let $\anotherpermutation\circ \asubset \defineequal \{\anotherpermutation \circ \gamma \;:\;\gamma \in \asubset\}$.
Note that if $\agroup$ is a subgroup of $\symmetricgroup_n$, $\anothergroup$ is a subgroup of $\agroup$, and $\anotherpermutation\in \agroup$,
then $\anotherpermutation \circ \anothergroup$ is a left coset of $\anothergroup$ in $\agroup$. 
The following proposition, which will be used in the proofs of
Lemma \ref{lemma:ContextFreeTradeoff} and Theorem \ref{theorem:MainTheoremEmbeddableGroup} 
follows from the fact that context-free languages are closed under homomorphisms.

\begin{proposition}
\label{proposition:RenamingProposition}
Let $\anothergroup\subseteq \symmetricgroup_n$, and $\apermutation$ be a permutation in $\symmetricgroup_n$. 
Let $\contextfreegrammar$ be a context-free grammar such that 
$\mylang(\contextfreegrammar) = \permutedcoordinates(\stringfication(\anothergroup),\apermutation)$. 
Then for each permutation $\anotherpermutation\in \symmetricgroup_n$ there is a context-free grammar 
$\contextfreegrammar_{\anotherpermutation}$
of size $|\contextfreegrammar_{\anotherpermutation}| = |\contextfreegrammar|$ generating 
$\permutedcoordinates(\stringfication(\anotherpermutation\circ \anothergroup),\apermutation)$.
\end{proposition}
\begin{proof}
Let $\contextfreegrammar = ([n],\variables,\grammarrelation,\avariable_1)$ 
	be a context free grammar generating $\permutedcoordinates(\stringfication(\anothergroup),\apermutation)$.
Let $\hat{\anotherpermutation}:[n]\cup \variables \rightarrow [n]\cup \variables$
be the extension of $\anotherpermutation$ to the set $[n]\cup \variables$ which 
sets $\hat{\anotherpermutation}(\asymbol) = \anotherpermutation(\asymbol)$ if $\asymbol\in [n]$
and $\hat{\anotherpermutation}(\asymbol)= \asymbol$ if $\asymbol \in \variables$. For each 
string $u = u_1...u_m \in ([n]\cup \variables)^*$ let 
$\hat{\anotherpermutation}(u) = \hat{\anotherpermutation}(u_1)...\hat{\anotherpermutation}(u_m)$.
Finally, let $\contextfreegrammar'$ be the context-free grammar obtained from 
$\contextfreegrammar$ by replacing each production rule $(\avariable,u)\in \productionrules$ 
with the production rule $(\avariable,\hat{\anotherpermutation}(u))$. Clearly, we have that 
$|\contextfreegrammar| = |\contextfreegrammar'|$. Additionally, it is straightforward
to verify that $\contextfreegrammar$ generates a string $a_1a_2...a_n\in [n]^n$ if and only
if $\contextfreegrammar'$ generates the string
$\anotherpermutation(a_1)\anotherpermutation(a_2)...\anotherpermutation(a_n)$. Therefore, 
$\mylang(\contextfreegrammar') = \permutedcoordinates(\stringfication(\anotherpermutation \circ  \anothergroup),\apermutation)$.
\end{proof}

The following theorem, which will be crucial to the proof of our main result (Theorem~\ref{theorem:MainTheoremEmbeddableGroup}),
upper bounds the symmetric grammar complexity of the automorphism group of a graph
in terms of the number of its vertices, its maximum degree, and its treewidth.
If the latter two quantities are bounded, then this upper bound is polynomial in the number of its vertices.
\begin{theorem}
\label{theorem:AutomorphismsTreewidth}
Let $\agraph$ be a connected graph with $n$ vertices, treewidth $\treewidth$ and maximum degree $\maximumdegree$. 
Then
$$\sgc(\automorphisms(\agraph)) \leq 2^{O(\treewidth \maximumdegree \log \maximumdegree)}\cdot n^{O(\treewidth)}.$$ 
Additionally, one can construct in time $2^{O(\treewidth \maximumdegree \log \maximumdegree)}\cdot n^{O(\treewidth)}$
a permutation $\alpha\in \symmetricgroup_n$ and a context-free grammar $\contextfreegrammar(\agraph)$ generating
	the language $\permutedcoordinates(\stringfication(\automorphisms(\agraph)),\apermutation)$.
\end{theorem}

\begin{myremark}
\label{remark:AutomorphismsPathwidth}
If the graph $\agraph$ of Theorem \ref{theorem:AutomorphismsTreewidth} has pathwidth $\pathwidth$, then 
 one may assume that  $\contextfreegrammar(\agraph)$ is a regular grammar. In other words, in this case, 
$\regsgc(\automorphisms(\agraph)) \leq 2^{O(\treewidth\maximumdegree\log\maximumdegree)}\cdot n^{O(\treewidth)}$. 
\end{myremark}

Theorem \ref{theorem:AutomorphismsTreewidth} can be simultaneously generalized in two ways. First, 
by allowing grammars to represent not only the automorphism group of a graph, but also groups that 
can be embedded in the graph. Second, not only the groups themselves but also left cosets of such 
groups can be represented in the same way. The result of these generalizations is stated in the next theorem.

\begin{theorem}
\label{theorem:MainTheoremEmbeddableGroup}
Let $\agroup\subgroup \symmetricgroup_n$, and suppose that $\agroup$ is embeddable on a graph $\agraph$
with $m$ vertices ($m\geq n$), maximum degree $\maximumdegree$, and treewidth $\treewidth$. Then, 
for each $\anotherpermutation\in \symmetricgroup_n$, 
$$\sgc(\anotherpermutation\circ \agroup)\leq 2^{O(\treewidth\maximumdegree\log\maximumdegree)}\cdot m^{O(\treewidth)}.$$
Additionally, given $\agraph$ and $\anotherpermutation$, one can construct in time 
$2^{O(\treewidth\maximumdegree\log\maximumdegree)}\cdot m^{O(\treewidth)}$ a permutation 
$\apermutation\in \symmetricgroup_n$ (depending only on $\agraph$) and a grammar 
$\contextfreegrammar_{\anotherpermutation}$ generating the language $\permutedcoordinates(\stringfication(\anotherpermutation\circ \agroup),\apermutation)$.
\end{theorem}

\begin{myremark}
\label{remark:MainTheoremEmbeddableGroupPathwidth}
If the graph $\agraph$ of Theorem \ref{theorem:MainTheoremEmbeddableGroup} has pathwidth $\pathwidth$, then
one may assume that $\contextfreegrammar_{\anotherpermutation}(\agraph)$ is a regular grammar. In other 
words, in this case, $\regsgc(\anotherpermutation\circ \agroup) 
\leq 2^{O(\treewidth\maximumdegree\log\maximumdegree)}\cdot m^{O(\treewidth)}$. 
\end{myremark}

\subsection{Proof of Theorem \ref{theorem:MainTheoremEmbeddableGroup}}
\label{subsection:ProofMainTheoremFormalLanguages}

In this section, we will prove Theorem \ref{theorem:MainTheoremEmbeddableGroup}, which establishes 
an upper bound for the symmetric grammar complexity of a permutation group $\agroup$ in function of the size,
treewidth and maximum degree of a graph embedding $\agroup$. On the way to prove 
Theorem \ref{theorem:MainTheoremEmbeddableGroup}, we will first prove Theorem \ref{theorem:AutomorphismsTreewidth}. 
The proofs of Remarks \ref{remark:AutomorphismsPathwidth} and \ref{remark:MainTheoremEmbeddableGroupPathwidth} follow 
by small adaptations of the proofs of Theorems \ref{theorem:AutomorphismsTreewidth} and \ref{theorem:MainTheoremEmbeddableGroup}
respectively.

\medskip
\noindent{\bf Subtree-Like Sets and Subterms.}
Let $r\in \Nplus$,  $\treelikeset\subseteq [r]^*$ be a tree-like set, 
$p,q \in \treelikeset$ and $u$ be the longest common prefix of $p$ and $q$. Let $p=up'$ and $q=uq'$. The {\em
distance} between $p$ and $q$ is defined as $|p'| + |q'|$.
We call a set $\subtreelikeset \subseteq \treelikeset$ {\em subtree-like} if
there exists a $p \in \subtreelikeset$, such that $p$ is a prefix of every 
$q\in \subtreelikeset$, and if the set $\subtreelikeset' = \{u
\mid pu \in \subtreelikeset\}$ is prefix-closed. In particular, for 
each $p\in \treelikeset$, the set $\treelikeset|_p$ is subtree-like.
One can obtain from $\subtreelikeset'$ a tree-like set $\subtreelikeset''$ by
making $\subtreelikeset'$ well-numbered in the obvious way. We call
$\subtreelikeset''$ the tree-like set {\em induced by} $\subtreelikeset$.
For a set $\treelikeset' \subseteq \treelikeset$, we call the smallest
subtree-like set containing $\treelikeset'$ the {\em closest ancestral closure
of $\treelikeset'$}.
For any subtree-like set $P \subseteq \positions(\aterm)$, we call $\aterm|_{P}$ a 
{\em subterm} of $\aterm$. If $P'$ is the induced tree-like set of $P$, then we 
call the corresponding term $\aterm'$ with $\positions(\aterm') = P'$ the 
{\em $\aterm$-term induced by $P$}.  For a position $p \in \positions(\aterm)$, 
we denote by $\aterm|_p$ the subterm of $\aterm$ rooted at $p$, i.e.~we 
let $\aterm|_p \defineequal \aterm|_{N|_p}$.

\medskip
\noindent{\bf Neighborhood of a Vertex, and Induced Subgraphs.}
Let $\agraph$ be a $n$-vertex graph. For a vertex $v \in [n]$, we let 
$\neighborhood(v) \defineequal \{u\in [n] \;:\; \{v, u\} \in E(\agraph)\}$
be the {\em neighborhood} of $v$. If $S\subseteq [n]$ then we let 
$\neighborhood(S) = \bigcup_{v\in S} \neighborhood(v)$ be the neighborhood
of $S$. Finally, we let $\closedneighborhood(S)= \neighborhood(S)\cup S$ 
be the {\em closed neighborhood} of $S$. The subgraph of $\agraph$ induced by 
$S$ is defined as $\agraph[S] = (S,\edgeset(\agraph[S]))$ where 
$E(\agraph[S]) = \edgeset(\agraph)\cap \binom{S}{2}$.

\medskip
\noindent{\bf Tree decomposition as Terms.}
If we regard the set $\binom{V(X)}{\le \decwidth+1}$ as an alphabet, then 
each width-$\treewidth$ tree decomposition of a graph $\agraph$ may be regarded as 
a term over $\binom{V(X)}{\le \decwidth+1}$. More precisely, 
let $\agraph$ be an $n$-vertex graph and  $\decwidth \in \{0,1,\ldots,n-1\}$. 
A {\em width-$\decwidth$ tree decomposition} (or simply {\em tree decomposition}, if $\decwidth$ is clear from the context)
of $\agraph$ is a term $\treedecomposition \in \myterms(\binom{\vertexset(\agraph)}{\le \decwidth + 1})$ satisfying the following axioms.
	\begin{enumerate}[label={(T\arabic*)}]
		\item $\bigcup_{\aposition \in \positions(\treedecomposition)} \treedecomposition(\aposition) = 
		\vertexset(X)$\label{def:tree:decomposition:vertices}
		\item For each vertex $\avertex \in \vertexset(\agraph)$ and each of its neighbors 
			$\anothervertex \in \neighborhood(\avertex)$, there is a position $\aposition \in \positions(\treedecomposition)$ 
			such that $\{\avertex, \anothervertex\} \subseteq \td(\aposition)$.\label{def:tree:decomposition:edges}
		\item For each vertex $\avertex \in \vertexset(\agraph)$, the set 
		$\{\aposition \in \positions(\treedecomposition) \mid \avertex \in \treedecomposition(\aposition)\}$ 
		induces a subterm of $\treedecomposition$.\label{def:tree:decomposition:subterm}
	\end{enumerate}
	The {\em treewidth} of $\agraph$, is defined as the smallest non-negative integer $\treewidth \in \bN$
	 such that $X$ admits a width-$\treewidth$ tree decomposition.

\medskip
\noindent{\bf Annotated Tree Decompositions.}
Let $\agraph$ be an $n$-vertex graph, $S$ and $S'$ be subsets of $[n]$ such that $|S|=|S'|$, and 
$\evtx:S\rightarrow S'$ be a bijection. We say that $\evtx$ is a 
{\em partial automorphism} of $\agraph$ if $\evtx$ is an isomorphism from the subgraph $\agraph[S]$ of $\agraph$
induced by $S$ to the subgraph $\agraph[S']$ of $\agraph$ induced by $S'$. 
Next, we define the notion of {\em annotated tree decomposition} of a graph $\agraph$. 
These are tree-decompositions whose bags are annotated with partial automorphisms.

\begin{definition}[Annotated Bags]
\label{definition:AnnotatedBag}
	Let $X$ be an $n$-vertex graph and $\treewidth \in \{0, \ldots, n-1\}$. A {\em $\treewidth$-annotated bag} is a pair 
	$(S, \evtx)$, where $S \in \binom{V(X)}{\le \treewidth+1}$, and $\evtx \colon \closedneighborhood[S] \to V(X)$
	is a function satisfying the following two properties. 
\begin{enumerate}
	\item\label{eq:annotated:bag:iso:immage} $\evtx(\closedneighborhood(S)) = \closedneighborhood(\evtx(S))$.  In other words, 
		the image of $\closedneighborhood(S)$ under $\evtx$ is equal to the closed neighborhood of the image of $S$
		under $\evtx$. 
	\item\label{eq:annotated:bag:iso} $\evtx$ is a partial automorphism of $\agraph$. 
\end{enumerate}	
\end{definition}

	We let $\annotatedbags(\agraph, \treewidth)$ be the set of all $\treewidth$-annotated bags of $\agraph$. If $b$ is a 
	$\treewidth$-annotated bag in $\annotatedbags(X,\treewidth)$, then we denote the first coordinate of $b$ by $b.S$
	and the second coordinate of $b$ by $b.\evtx$. In other words, $b = (b.S,b.\evtx)$. 
	We let $\poormap \colon \annotatedbags(X,\treewidth) \to \binom{V(X)}{\le w+1}$ be the map that takes an
	annotated bag $b\in \annotatedbags(X,\treewidth)$ and sends it to the bag $\poormap(b) = b.S \in \binom{V(X)}{\le k+1}$.
	In other words, the map $\poormap$ erases the second coordinate of the annotated bag $b$. 
	We extend $\poormap$ to terms in $\myterms(\annotatedbags(\agraph,\treewidth))$ positionwise. 
	More precisely, for each term $\eterm \in \myterms(\annotatedbags(X,\treewidth))$, we let $\poormap(\annotatedterm)$ be the 
	term in $\myterms(\binom{V(X)}{\le k+1})$ where $\positions(\poormap(\eterm)) \defineequal \positions(\annotatedterm)$ and 
	 $\poormap(\annotatedterm)(p)  \defineequal \poormap(\eterm(p))$ for each $p \in \positions(\aterm)$. 
	 We say that a term $\eterm\in \myterms(\annotatedbags(\agraph,\treewidth))$ is an {\em annotation} of a term 
	 $\aterm \in \myterms(\binom{V(X)}{\le k+1})$ if $\poormap(\eterm) = t$. Note that a term $t\in \myterms(\binom{V(X)}{\le k+1})$
	 may have many annotations.

In Definition \ref{definition:AnnotatedBag}, once 
a subset $S\subseteq V(\agraph)$ is fixed, 
there are at most $\cO(n^{k+1})$ choices for the image of $S$ under the partial isomorphism $\evtx$. 
Once such an image is fixed, for each vertex $x \in S$
there are at most $\maximumdegree!$ ways of mapping the neighbors of $x$ to the
neighbors of $\evtx(x)$. Hence there are at most $(\maximumdegree!)^{k+1}$
choices for obtaining a partial automorphism for a fixed image of $S$. Therefore, 
by noting that $\maximumdegree! = 2^{O(\maximumdegree\log \maximumdegree)}$, we have
the following observation.

\begin{observation}
	Let $X$ be a graph of maximum degree $\maximumdegree$ and let $k \in \{0,\ldots,n-1\}$. Then, 
	$\card{\annotatedbags(X,\treewidth)} \le 2^{\cO(k \maximumdegree \cdot \log \maximumdegree)}\cdot n^{\cO(\treewidth)}$.
\end{observation}

\begin{definition}[Annotated Tree Decomposition]
\label{definition:AnnotatedTreeDecomposition}
	Let $\annotatedtreedecomposition$ be a term in $\myterms(\annotatedbags(X,\treewidth))$.
	We say that $\annotatedtreedecomposition$ is an {\em annotated width-$\treewidth$ tree decomposition} if the following 
	conditions are satisfied. 
	\begin{enumerate}
	\item $\poormap(\etd)$ is a tree decomposition. 
	\item \label{eq:enrichmed:td:consistency}
		  for each $p \in \positions(\annotatedtreedecomposition)$ with children $p1, \ldots, pd$, and for each 
		$j\in [d]$, the restriction of $\etd(p).\evtx$ to $N[\etd(p).S] \cap N[\etd(pj).S]$ is equal to the restriction 
			of $\etd(pj).\evtx$ to $N[\etd(p).S] \cap N[\etd(pj).S]$.
	\end{enumerate} 
\end{definition}

Intuitively, the first condition states that if we take an annotated tree decomposition $\annotatedtreedecomposition$
and forget annotation then the result is a tree-decomposition of $\agraph$. The second condition guarantees that the 
annotation is consistent along the whole tree decomposition, in the sense that for each vertex $x\in\vertexset(\agraph)$,
if the partial automorphism of one bag sends $x$ to vertex $x'$, then the partial automorphism of each bag sends 
$x$ to $x'$. Each annotated tree decomposition $\etd$ gives rise to a map 
$\eaut(\etd) \colon V(X) \to V(X)$ which sets $\eaut(\etd)|_{N[\etd(p).S]} = \etd(p).\evtx$ for each  $p \in \positions(\etd)$. 
We call the map $\eaut$ the {\em annotation morphism} of $\annotatedtreedecomposition$. The following lemma is the 
main technical tool of this section.

\begin{lemma}
\label{lem:annotation:automorphism}
Let $X$ be an $n$-vertex graph of treewidth $\treewidth$ %, $\td$ a width-$\treewidth$ tree decomposition of $X$ 
and $\aperm \in \sym_n$. Then, $\aperm$ is an automorphism of $X$ if and only if there exists an annotated tree 
decomposition $\etd$ of $X$ such that $\aperm = \eaut(\etd)$.
\end{lemma}
\begin{proof} ({\em Only if direction.}) First, we show that if $\aperm$ is an automorphism of 
$\agraph$ then there is an annotated tree decomposition $\annotatedtreedecomposition$ such that
 $\aperm = \eaut(\etd)$.  Let $\td$ be a width-$\treewidth$ tree decomposition of
$\agraph$ and suppose $\aperm$ is an automorphism. Then, we construct an annotated
tree decomposition $\etd$ with $\poormap(\etd) = \td$ by letting for each $p
\in \positions(\etd)$, $\etd(p).S \defineequal \td(p)$ and $\etd(p).\evtx \defineequal
\aperm|_{\closedneighborhood[\etd(p).S]}$. In other words, the annotation of each bag is simply the restriction of 
$\aperm$ to the closed neighborhood of $\etd(p).S$. Clearly, by construction we have that $\aperm = \eaut(\etd)$.
Therefore, it is enough to show that $\annotatedtreedecomposition$ is an annotated tree decomposition. 
 Since $\aperm$ is an automorphism, one immediately
verifies that for each $p \in \positions(\etd)$, $\evtx_p \defineequal  \etd(p).\evtx$
is an isomorphism from $X[\closedneighborhood[S]]$ to $X[\closedneighborhood[\evtx_p(S)]]$, i.e.~condition
(\ref{eq:annotated:bag:iso}) of Definition \ref{definition:AnnotatedBag} holds.
Condition (\ref{eq:enrichmed:td:consistency}) of
Definition~\ref{definition:AnnotatedTreeDecomposition} is satisfied as well,
since by construction, for any pair of positions $p, p' \in \positions(\etd)$
and $x \in \etd(p).S \cap \etd(p').S$, $\etd(p).\evtx(x) = \aperm(x) =
\etd(p').\evtx(x)$. 
	
	({\em If direction}) Suppose that there is an annotated width-$\treewidth$ tree
decomposition $\etd$ of $X$ such that $\aperm = \eaut(\etd)$. 
	We show that $\aperm$ is an automorphism of $X$. First, we argue that
$\aperm$ is well-defined. Since $\poormap(\etd)$ is a tree decomposition, we
have by~\ref{def:tree:decomposition:vertices} that for each $x \in V(X)$, there
is at least one position $p \in \positions(\etd)$ such that $x \in \etd(p).S$.
By property~\ref{def:tree:decomposition:subterm} of tree decompositions and
condition (\ref{eq:enrichmed:td:consistency}) of Definition~\ref{definition:AnnotatedTreeDecomposition}, 
we can conclude that~$\aperm(x)$ is assigned a unique value, and therefore that $\aperm$ is a well defined function. It
remains to argue that $\aperm$ is surjective, injective, and indeed an automorphism.
\begin{claim}\label{claim:two:one} 
Let $x, x', y' \in V(X)$. If
$\aperm(x) = x'$ and $\{x', y'\} \in E(X)$, then there exists a $y \in V(X)$
such that $\{x, y\} \in E(X)$ and  $\aperm(y) = y'$.  \end{claim}
\begin{claimproof} If $\aperm(x) = x'$, then there exists some position $p \in
\positions(\etd)$ such that $\etd(p).\evtx(x) = x'$. Now, Condition (\ref{eq:annotated:bag:iso})
of Definition \ref{definition:AnnotatedBag} implies that $\etd(p).\evtx|_{N[x]}$ is a bijection from
$N[x]$ to $N[x']$. Since, by assumption, $y' \in N[x']$, there is a vertex $y
\in N[x]$ such that $\etd(p).\evtx(y) = y'$, and hence $\aperm(y) = y'$.
\end{claimproof} 

From Claim~\ref{claim:two:one}, the following claim follows straightforwardly by induction 
on the length of paths. 

\begin{claim}\label{claim:Paths}
Let $x,x',y'\in V(X)$. If $\aperm(x) = x'$ and there exists a path from 
$x'$ to $y'$, then there exists a vertex $y\in V(X)$ such that $\aperm(y) = y'$.
\end{claim}
\begin{claimproof}
The proof is by induction on the length of paths. In the base case, the path has length
$0$. In this case, $x'=y'$ and the claim follows trivially by setting $y=x$. Now, let 
$r\geq 0$, and assume that the claim is true for every path of length at most 
$r$. Let $x'x_1'x_2'....x_r'y'$ be a path of length $r+1$ from $x'$ to $y'$. 
Then, by the induction hypothesis, there exists $x_r\in V(X)$ such that $\aperm(x_r) = x_r'$. 
Now since $\{x_r',y'\}$ is an edge in $E(X)$, by Claim \ref{claim:two:one}, 
we have that there exists a $y\in V(X)$ such that $\aperm(y) = y'$.  
\end{claimproof}

Since $\agraph$ has at least one vertex, we have that there exist vertices 
$x_0, x_0' \in V(X)$ such that $\aperm(x_0) = x_0'$. Now, since $X$ is connected, 
there is a path from $x_0'$ to any other vertex $y$ in $X$. Therefore, from Claim \ref{claim:Paths} and 
from the fact that $\aperm$ is well-defined, we can conclude that $\aperm$ is surjective.
	
Since $\aperm$ is a surjective map whose domain and codomain have the
same size, we can infer that $\aperm$ is injective as well.

	Now, suppose $\{x, y\} \in E(X)$. By property
\ref{def:tree:decomposition:edges} of tree decompositions, there is a position
$p \in \positions(\etd)$ such that $\{x, y\} \subseteq \etd(p).S$. Let $S
\defineequal \etd(p).S$ and $\evtx \defineequal  \etd(p).\evtx$.  By condition
(\ref{eq:annotated:bag:iso}) of Definition \ref{definition:AnnotatedBag}, 
$\evtx$ is an isomorphism from $X[\closedneighborhood[S]]$ to $X[\closedneighborhood[\evtx(S)]]$, 
so we know that $\{\evtx(x), \evtx(y)\} \in E(X)$. Since
$\evtx(x) = \aperm(x)$ and $\evtx(y) = \aperm(y)$, we have that $\{\aperm(x),
\aperm(y)\} \in E(X)$.  On the other hand, if $\{\aperm(x), \aperm(y)\} \in
E(X)$, then by Claim~\ref{claim:two:one}, $\{x, y\} \in E(X)$. 
	
	This concludes the proof	that $\aperm$ is an automorphism.
\end{proof}

\begin{definition}
\label{definition:permutationYield}
A tree decomposition $\treedecomposition$ is called {\em permutation yielding}, if there is a 
bijection $\cLang \colon \leaves(\positions(\treedecomposition)) \to V(X)$ such that for each leaf 
$p \in \leaves(\positions(\treedecomposition))$,  $\treedecomposition(p) = \{\cLang(p)\}$.
\end{definition}

In other words, a tree decomposition $\treedecomposition$ is permutation yielding if each vertex occurs
in precisely one leaf bag. The next lemma shows that any tree decomposition $\treedecomposition$ can be transformed in
polynomial time into a permutation yielding tree decomposition of same width. 
We note that a statement analogous to Lemma~\ref{lemma:yielding:tree:decomposition} can also 
be obtained by observing that tree-decompositions can be converted in polynomial time into branch decompositions 
of roughly the same width \cite{RS91}.  
We include a proof of Lemma \ref{lemma:yielding:tree:decomposition} for completeness. 
\begin{lemma}\label{lemma:yielding:tree:decomposition} Let $X$ be an $n$-vertex
graph, $k \in \{0,\ldots,n-1\}$, and $\treedecomposition$ a width-$\decwidth$
tree decomposition of $X$. Then, one can construct from $\treedecomposition$ in
polynomial time a permutation yielding width-$\decwidth$ tree decomposition.
\end{lemma} 
\begin{proof} For each position $p\in
\positions(\treedecomposition)$, let $\arity(p)$ be the number of children of
$\aposition$ in $\positions(\treedecomposition)$.  We can assume that for each
$v \in V(X)$, there is at least one position $\leafposition \in
\leaves(\positions(\treedecomposition))$ such that $\treedecomposition(\leafposition) =
\{v\}$, otherwise we could pick an arbitrary position $p \in
\positions(\treedecomposition)$ with $v \in \treedecomposition(p)$ and add a
leaf $p' = p\cdot (\arity(p)+1)$ and set $\treedecomposition(p') = \{v\}$. 
Hence, we can assume that there is a set
$S \subseteq \leaves(\positions(\treedecomposition))$ containing precisely one
leaf $\leafposition_v \in S$ with $\treedecomposition(\leafposition_v) = \{v\}$ for each $v \in
V(X)$. Let $S^+$ be the closest ancestral closure of $S$ in
$\positions(\treedecomposition)$. We let $\treedecomposition'$ be the
$\treedecomposition$-term induced by the closest ancestral closure of $S$. Note that 
$S = \leaves(\positions(\treedecomposition|_{S^+}))$. 
	
By construction, there is a bijection $\cLang \colon
\leaves(\positions(\treedecomposition|_{S^+})) \to V(X)$ with
$\treedecomposition(p) = \{\cLang(p)\}$ for each $p \in
\leaves(\positions(\treedecomposition|_{S^+}))$. Let $p^* \in
\positions(\treedecomposition) \setminus \positions(\treedecomposition|_{S^+})$
and consider $x, y \in \treedecomposition(p^*)$ (where possibly $x = y$). Let
$q^* \in S^+$ denote the position that is closest to $p^*$ among all positions
in $S^+$.  Since there are leaves $\leafposition_x, \leafposition_y \in S$ such that
$\treedecomposition(\leafposition_x) = \{x\}$ and $\treedecomposition(\leafposition_y) = \{y\}$
we can conclude by Property~\ref{def:tree:decomposition:subterm} of tree
decompositions that $\{x, y\} \subseteq \treedecomposition(q^*)$. We have
argued that $\treedecomposition|_{S^+}$ satisfies
Properties~\ref{def:tree:decomposition:vertices}
and~\ref{def:tree:decomposition:edges} of tree decompositions, and we observe
that property~\ref{def:tree:decomposition:subterm} remains intact on
$\treedecomposition|_{S^+}$ as well. 
	
	Hence, the $\treedecomposition$-term induced by $S^+$ is a permutation
yielding width-$\decwidth$ tree decomposition of $X$ and it is clear that the
above construction can be implemented in polynomial time in the size of $\treedecomposition$. \end{proof}

\newcommand{\vars}{\variables} \newcommand{\iv}{\initialvariable}
\newcommand{\gr}{\grammarrelation} \newcommand{\pr}{\gr}
\newcommand\specialvariable{\mathfrak{b}} \newcommand\svar{\specialvariable}

Let $\etd$ be an annotated tree decomposition with $r$ leaves,
and let $\yield(\etd) = (S_1,\evtx_1) \ldots (S_r,\evtx_r)$ be 
the yield of $\annotatedtreedecomposition$. 
In other words, $\yield(\etd)$ is the sequence
of annotated bags obtained by reading the leaves of
$\annotatedtreedecomposition$ from left to right.  Then we define the {\em
annotation yield} of $\etd$ as the sequence $\eyield(\etd) \defineequal
\evtx_1(S_1) \ldots \evtx_m(S_r)$. Note that if $\td$ is a permutation 
yielding tree decomposition of $\agraph$, and $\etd$ is an 
annotation of $\td$, then $r=|\vertexset(\agraph)|$, and $\eyield(\etd)$ 
is a string of singletons of the form $\{v_1\}\{v_2\}\dots\{v_{r}\}$ where 
$v_i\in \vertexset(\agraph)$ for each $i\in [r]$, and $v_i\neq v_j$ for $i\neq j$.

\begin{retheorem}[\ref{theorem:AutomorphismsTreewidth}]
\label{retheorem:AutomorphismsTreewidth}
Let $\agraph$ be a graph with $n$ vertices, treewidth $\treewidth$ and maximum degree $\maximumdegree$. 
Then $\sgc(\automorphisms(\agraph)) \leq 2^{O(\treewidth \maximumdegree \log \maximumdegree)}\cdot n^{O(\treewidth)}$. 
Additionally, one can construct in time $2^{O(\treewidth \maximumdegree \log \maximumdegree)}\cdot n^{O(\treewidth)}$ a permutation $\alpha$ and 
a context-free grammar $\contextfreegrammar(\agraph)$ generating $\permutedcoordinates(\stringfication(\automorphisms(\agraph)),\apermutation)$.
\end{retheorem}
\begin{proof}
Since the graph $\agraph$ has treewidth $\treewidth$, one can construct in time 
$2^{O(\treewidth)}\cdot n^{O(1)}$ a width $O(\treewidth)$ tree decomposition $\treedecomposition$ of $\agraph$.
Additionally, from Lemma \ref{lemma:yielding:tree:decomposition}, one can assume that 
$\treedecomposition$ is yielding. Let $\yield(\treedecomposition) = \{v_1\}\{v_2\}...\{v_n\}$. Then we let 
$\alpha_{\treedecomposition}$ be the permutation in $\symmetricgroup_n$ with 
$\stringfication(\alpha_{\treedecomposition}) = v_1v_2...v_n$. We set $\alpha = \alpha_{\treedecomposition}^{-1}$. 
Since $\treedecomposition$ can be constructed in time 
$2^{O(\treewidth)}\cdot n^{O(1)}$, so can the permutation $\alpha$. 

We show that from $\treedecomposition$ one can construct 
a context-free grammar $\contextfreegrammar$ accepting the language 
	$\mylang(\contextfreegrammar) = \permutedcoordinates(\stringfy{\automorphisms(\agraph)},\apermutation)$. 
Intuitively, the parse trees accepted by the grammar $\contextfreegrammar$ correspond to annotations 
of $\treedecomposition$, and by Lemma \ref{lem:annotation:automorphism}, 
these annotations correspond to automorphisms of $\agraph$. 
Formally, the grammar $\cfg = (\alphabet, \vars, \gr, \iv)$ is defined as follows. 
We let $\alphabet = V(X) = [n]$ and $\vars = \positions(\treedecomposition)\times (\annotatedbags(X, \treewidth) \cup \{\iv\})$ where $\iv$ is the 
initial variable of $\cfg$.
Recall that $\poormap \colon \annotatedbags(X, \treewidth) \to \binom{V(X)}{\le \treewidth+1}$ is the map that erases the second
coordinate from each annotated bag $b\in \annotatedbags$. The set $\gr$ contains the following rules.

\begin{enumerate}
	\item A rule $B_1\rightarrow (\lambda,b)$ for each annotated bag $b\in\annotatedbags(\agraph,\treewidth)$ such that 
		$\rho(b)= \treedecomposition(\lambda)$. Intuitively, each such $b$ is an annotated bag corresponding
		to the bag at the root of $\treedecomposition$. 
	\item For each non-leaf position $p \in \positions(\td)\setminus \leaves(\positions(\td))$, 
		with children $p1, \ldots, pd$, we have a rule $$(p,b) \to (p1,b_1) (p2, b_2) \ldots (pd,b_d),$$ for 
		each sequence $b,b_1,...,b_d$ of annotated bags in $\annotatedbags(\agraph,\treewidth)$ satisfying the following
		conditions:  	
	\begin{enumerate}[label={(\roman*)}]  
		\item $\poormap(b) = \td(p)$ and for $j \in [d]$, $\poormap(b_j) = \td(pj)$, and 
		\item \label{enum:aut:grammar:consistency} for each $j \in [d]$, $b.\evtx|_{S^*} = b_j.\evtx|_{S^*}$ where $S^* = b.S \cap b_j.S$. \end{enumerate}
	\item A rule $(\aposition,b)\rightarrow j$ for each leaf position $\aposition\in \positions(\treedecomposition)$ with 
		$b.S = \{i\}$ and $b.\evtx(i) = j$. 
\end{enumerate}

These rules defined above ensure that if we take an accepting parse tree $\parsetree$ of $\contextfreegrammar$ and 
remove its root (i.e the variable $B_1$) and its leaves (which are labeled with numbers in $[n]$) then we are left 
with an annotated version $\annotatedtreedecomposition$ of the tree decomposition $\treedecomposition$. By Lemma 
\ref{lem:annotation:automorphism}, $\annotatedtreedecomposition$ is an annotation of $\treedecomposition$ if and only
if the map $\eaut(\annotatedtreedecomposition):\vertexset(\agraph)\rightarrow \vertexset(\agraph)$ is an automorphism 
of $\agraph$. Therefore, since $\stringfication(\eaut(\annotatedtreedecomposition)) = \yield(\parsetree)$, 
we have that $\mylang(\contextfreegrammar)  = \permutedcoordinates(\stringfy{\automorphisms(\agraph)},\apermutation)$. 

Since we can assume that
$\card{\positions(\td)} = \cO(\treewidth n)$ (see e.g.~\cite[Lemma 7.4]{CyganEtAl15}),
and for each bag, there are at most $2^{\cO(k\maximumdegree \log
\maximumdegree)} \cdot n^{\cO(\treewidth )}$ annotations, we have that $\card{\cfg} = 2^{\cO(\treewidth
\maximumdegree \log \maximumdegree)} \cdot n^{\cO(\treewidth)}$, as claimed.  \end{proof}

\begin{reremark}[\ref{remark:AutomorphismsPathwidth}]
\label{reremark:AutomorphismsPathwidth}
If the graph $\agraph$ of Theorem \ref{theorem:AutomorphismsTreewidth} has pathwidth $\pathwidth$, then 
 one may assume that  $\contextfreegrammar(\agraph)$ is a regular grammar. In other words, in this case, 
$\regsgc(\automorphisms(\agraph)) \leq 2^{O(\treewidth\maximumdegree\log\maximumdegree)}\cdot n^{O(\treewidth)}$. 
\end{reremark}
\begin{proof}
Any $n$-vertex graph of pathwidth $\pathwidth$ admits a path decomposition $\mathcal{D}$ in which each vertex is introduced exactly
once. Therefore, if we apply the construction that transforms tree decompositions in yielding tree decompositions 
stated in Lemma \ref{lemma:yielding:tree:decomposition}, then we actually obtain a path decomposition where each node 
has a child of arity $1$ and a child of arity $0$, except for the internal node farthest away from the root, which has two
children of arity $0$ (one of the children of arity $0$ of each bag is a singleton containing the introduced vertex). 
Therefore, when applying the construction in Theorem \ref{theorem:AutomorphismsTreewidth}, the resulting grammar is regular. 
\end{proof}

\begin{retheorem}[\ref{theorem:MainTheoremEmbeddableGroup}] 
\label{retheorem:MainTheoremEmbeddableGroup}
Let $\agroup\subgroup \symmetricgroup_n$, and suppose that $\agroup$ is embeddable on a graph $\agraph$
with $m$ vertices ($m\geq n$), maximum degree $\maximumdegree$, and treewidth $\treewidth$. Then, 
for each $\anotherpermutation\in \symmetricgroup_n$, 
$$\sgc(\anotherpermutation\circ \agroup)\leq 2^{O(\treewidth\maximumdegree\log\maximumdegree)}\cdot m^{O(\treewidth)}.$$
Additionally, given $\agraph$ and $\anotherpermutation$, one can construct in time 
$2^{O(\treewidth\maximumdegree\log\maximumdegree)}\cdot m^{O(\treewidth)}$ a permutation 
$\apermutation\in \symmetricgroup_n$ (depending only on $\agraph$) and a grammar 
$\contextfreegrammar_{\anotherpermutation}$ generating the language $\permutedcoordinates(\stringfication(\anotherpermutation\circ \agroup),\apermutation)$.
\end{retheorem}
\begin{proof}
This is a consequence of Theorem \ref{theorem:AutomorphismsTreewidth}, together with the fact that 
context free grammars are closed under homomorphisms. More precisely, we first construct  in time 
$2^{O(\treewidth\maximumdegree\log\maximumdegree)}\cdot m^{O(\treewidth)}$ a permutation $\apermutation'\in \symmetricgroup_m$, 
and a context-free grammar $\contextfreegrammar'$ such that 
$\mylang(\contextfreegrammar) = \permutedcoordinates(\stringfy{\automorphisms(\agraph)},\apermutation)$. 
Now let $h:[m]\backslash [n]\rightarrow \{\emptysymbol\}$ be the map that sends each number in  $[m]\backslash [n]$ to the empty symbol $\varepsilon$. 
Then using $\contextfreegrammar$, one can construct in time polynomial in $|\contextfreegrammar'|$ a context-free grammar $\contextfreegrammar''$ 
whose language $\mylang(\contextfreegrammar'')$ is the homomorphic image of $\mylang(\contextfreegrammar')$ under $h$. Additionally, 
one may assume that the grammar $\contextfreegrammar''$ has no production rule containing the empty-symbol $\varepsilon$. Let 
$\apermutation = \apermutation'|_{[n]}$ be the permutation in $\symmetricgroup_n$ obtained by restricting $\apermutation'$ to $[n]$.
Note that $\apermutation$ is well defined, since the fact that $\agroup\subgroup \symmetricgroup_n$ is embeddable in $\agraph$ implies 
that $\apermutation'([n]) = [n]$, and therefore that $\apermutation([n])=[n]$. Then we have that the language accepted by $\contextfreegrammar''$ is 
$\mylang(\contextfreegrammar'') = \permutedcoordinates(\stringfication(\automorphisms(\agraph)),\apermutation)$. 

Finally, let $\anotherpermutation:[n]\rightarrow [n]$ be a permutation in $\symmetricgroup_n$. 
Then we can regard $\anotherpermutation$ as a usual map from $[n]$ to $[n]$, and 
using again the fact that context-free languages are closed under homomorphism, we can construct in time 
	$O(|\contextfreegrammar''|)$ a context-free grammar $\contextfreegrammar$ accepting the homomorphic
	image of $\mylang(\contextfreegrammar'')$ under $\anotherpermutation$. This homomorphic image is 
	simply the language  $\permutedcoordinates(\stringfy{\anotherpermutation\circ \agroup},\apermutation)$. 
\end{proof}

\begin{reremark}[\ref{remark:MainTheoremEmbeddableGroupPathwidth}]
\label{reremark:MainTheoremEmbeddableGroupPathwidth}
If the graph $\agraph$ of Theorem \ref{theorem:MainTheoremEmbeddableGroup} has pathwidth $\pathwidth$, then
one may assume that $\contextfreegrammar_{\anotherpermutation}(\agraph)$ is a regular grammar. In other 
words, in this case, $\regsgc(\anotherpermutation\circ \automorphisms(\agraph)) 
\leq 2^{O(\treewidth\maximumdegree\log\maximumdegree)}\cdot m^{O(\treewidth)}$. 
\end{reremark}
\begin{proof}
The proof of this remark follows from Theorem \ref{theorem:MainTheoremEmbeddableGroup} by using an 
argument analogous to the one used in the proof of Remark \ref{remark:AutomorphismsPathwidth}. 
\end{proof}

\section{Polytopes for Permutation Groups}
\label{section:PolytopesForPermutationGroups}

In linear-programming theory, the $n$-permutahedron is the polytope 
$\apolytope(\symmetricgroup_n)$ formed by the convex-hull of the set of
permutations of the set $\{1,\dots,n\}$. It can be shown that to define the 
permutahedron on the $n$-dimensional space, $2^{\Omega(n)}$ constraints are 
required. On the other hand, a celebrated result from Goemans states that 
the $n$-permutahedron has extended formulations with $O(n\log n)$ variables
and constraints \cite{Goemans2015smallest}. 

More generally, given a subgroup $\agroup \subgroup \symmetricgroup_n$, one can define the 
$\agroup$-hedron as the convex-hull of the permutations in $\agroup$. The technique used in 
\cite{Goemans2015smallest} to upper bound the extension complexity of polytope $\apolytope(\symmetricgroup_n)$, 
which is based on the existence of sorting networks of size $O(n\log n)$ \cite{Ajtai1983sorting}, 
has been used to show that polytopes corresponding to certain families of groups have small extension complexity. 
This includes polytopes corresponding to the alternating group \cite{Weltge2012erweiterte}, and to finite reflection 
groups \cite{kaibel2010branched,kaibel2013constructing,Humphreys1992reflection,BenTalNemirovski2001polyhedral}. 
Nevertheless, techniques to prove non-trivial upper bounds on the extension complexity of polytopes associated with
general permutation groups based on structural properties of these groups are still lacking. We note that a trivial upper bound of 
$|\agroup|$ can be obtained from the fact that the extension complexity of a polytope is upper bounded by 
its number of vertices. Nevertheless, $|\agroup|$ may have up to $n!= 2^{\Omega(n\log n)}$ elements. 

In this section, by combining our main theorem (Theorem \ref{theorem:MainTheoremEmbeddableGroup}) 
with a connection established in \cite{Pesant2009} between the grammar complexity of a given 
formal language $\lang\subseteq [n]^r$ (for $n,r\in \Nplus$) and the extension complexity of 
the polytope $\apolytope(\lang)$ associated with $\lang$, we obtain a new approach for proving upper bounds on 
the extension complexity of a general permutation group $\agroup \subgroup \symmetricgroup_n$ 
based on structural parameters of graphs embedding $\agroup$ (Theorem \ref{theorem:EmbeddabilityVsPolytopes}). 
We note that Theorem \ref{theorem:EmbeddabilityVsPolytopes} is more general in the sense that it also can be used to upper bound the extension complexity of 
polytopes associated with cosets of $\agroup$. 

Let $\realvariables$ be a set of real variables. A {\em real vector} over $\realvariables$ is a function 
$\avector:\realvariables\rightarrow \R$. We let $\reals^{\realvariables}$ be the set of all real vectors over $\realvariables$. Given a set 
$\setvectors = \{\avector_1,\dots,\avector_r\}$ of real vectors, the convex-hull of $\setvectors$ is the set 
$$\convexhull(\setvectors) = \{\sum_{i=1}^k \alpha_i \avector_i \;:\; \alpha_i\geq 0, \sum_{i=1}^k\alpha_i = 1\}$$ 
of all convex linear-combinations of vectors in $\setvectors$. A subset $\apolytope\subseteq \reals^{\realvariables}$ is a {\em polytope over $\realvariables$} if 
$\apolytope = \convexhull(\setvectors)$ for some finite set $\setvectors$ of real vectors over $\realvariables$. 
For each such a polytope $\apolytope$, there is a finite 
set $\setinequalities$ of linear inequalities over $\realvariables$ such that $\apolytope$ is the set of vectors in $\reals^{\realvariables}$
which satisfy each inequality in $\setinequalities$. 

Let $\realvariables$ and $\otherrealvariables$ be sets of real variables with 
$\realvariables\cap \otherrealvariables = \emptyset$. We say that a $(\realvariables \cup \otherrealvariables)$-polytope 
$\anotherpolytope$ is an {\em extended formulation} of $\apolytope$ if there exists a linear projection 
$\rho:\reals^{\realvariables\cup \otherrealvariables}\rightarrow \reals^{\realvariables}$
such that $\apolytope = \rho(\anotherpolytope)$. The {\em extension complexity of $\apolytope$}, denoted by $\extensioncomplexity(\apolytope)$,
is defined as the least number of inequalities necessary to define an extended formulation of $\apolytope$. 

For each $n\in \Nplus$, we let $[n]^{\realvariables}$ be the set of real vectors over $\realvariables$ whose coordinates are chosen from the 
set $[n]$. For $r\in \Nplus$, let $w=w_1\dots w_r$ be a string in $[n]^r$, and let $\mathcal{X}_r = \{x_1,...,x_r\}$ be an ordered set of real variables. 
We let $\hat{w}:\realvariables_r\rightarrow [n]$ be the real vector over $\realvariables_r$ which sets $\hat{w}_i = w_i$ for each $i\in [r]$. 
Given a subset $\lang\subseteq [n]^r$, the $\realvariables_r$-polytope associated with $\lang$ is defined as 
$\polytope(\lang) = \convexhull(\{\hat{w}\;:\; w\in \lang\})$. 

The following theorem, proved in \cite{Pesant2009}, relates 
the grammar complexity of a subset $\lang\subseteq [n]$ with the extension complexity of the polytope $\polytope(\lang)$. 

\begin{theorem}[\cite{Pesant2009}]
	\label{theorem:Pesant}
	Let $\grammar$ be a context-free grammar such that $\mylang(\grammar) \subseteq [n]^{r}$ for some $n,r\in \Nplus$. 
	Then the extension complexity of the polytope $\apolytope(\mylang(\grammar))$ is upper bounded by $|\grammar|^{O(1)}$.
	Additionally, a system of inequalities defining $\apolytope(\mylang(\grammar))$ can be constructed in time 
	$|\grammar|^{O(1)}$. 
\end{theorem}

If $\agroup$ is a subgroup of $\symmetricgroup_n$, and $\anotherpermutation\in \agroup$, then we let 
$$\apolytope(\anotherpermutation\circ \agroup) \defeq \apolytope(\stringfication(\anotherpermutation\circ \agroup))$$ be the 
polytope associated with the coset $\anotherpermutation\circ \agroup$.
The following theorem, which is the main result of this section, follows by a direct combination of
Theorem \ref{theorem:MainTheoremEmbeddableGroup} with Theorems \ref{theorem:Pesant}. 

\begin{theorem}
\label{theorem:EmbeddabilityVsPolytopes}
Let $\agroup\subgroup \symmetricgroup_n$, and suppose that $\agroup$ is embeddable on a graph $\agraph$
with $m$ vertices ($m\geq n$), maximum degree $\maximumdegree$, and treewidth $\treewidth$. Then, 
for each $\anotherpermutation\in \symmetricgroup_n$, the extension complexity of the polytope
$\apolytope(\anotherpermutation\circ \agroup)$ is at most $2^{O(\treewidth\maximumdegree\log\maximumdegree)}\cdot m^{O(\treewidth)}.$
Additionally, given $\agraph$ and $\anotherpermutation$, a system of inequalities defining $\apolytope(\anotherpermutation\circ \agroup)$ 
can be constructed in time $2^{O(\treewidth\maximumdegree\log\maximumdegree)}\cdot m^{O(\treewidth)}$. 
\end{theorem}

\section{Complexity Theoretic Tradeoffs}
\label{section:RelationsComplexityMeasures}

In 1969 Babai and Bouwer showed independently that any subgroup $\agroup$ of $\symmetricgroup_n$ 
can be embedded in a connected graph with $O(n+|G|)$ vertices \cite{Babai1969representation,Bouwer1969section}.
Note that $|G|$ can be as large as $n!$. Classifying which groups can, or cannot, be embedded on 
connected graphs with a much smaller number of vertices is an important problem in algebraic graph theory \cite{Babai1995automorphism}. 
Indeed, constructing an explicit class of graphs with superpolynomial graph embedding complexity is still an 
open problem, although a conjecture of Babai states that the alternating group $\alternatinggroup_n$ has graph embedding
complexity $2^{\Omega(n)}$ \cite{Babai1981abstract}. We note that Liebeck has shown that any graph whose 
automorphism group is {\em isomorphic} to the alternating group $\alternatinggroup_n$ (as an abstract group) 
has an exponential number of vertices \cite{Liebeck1983graphs}. Nevertheless this result does not extend
to the graph embedding setting. 

In this section we use our main theorem to establish a trade-off
between the index of a subgroup $\agroup$ of $\symmetricgroup_n$, and structural parameters of graphs
embedding $\agroup$. In particular, for several classes of graphs $\graphclass$, 
this trade-off can be used to prove lower bounds on the $\graphclass$-embedding complexity 
of subgroups of $\symmetricgroup_n$ of small index (i.e index up to $2^{cn}$ for some small constant $c$).
We start by stating the following immediate observation. 

\begin{observation}
\label{observation:UnionClosure}
Let $\contextfreegrammar_1$ and $\contextfreegrammar_1$ be context-free grammars. 
Then there is a context-free grammar 
$\contextfreegrammar_1\cup \contextfreegrammar_2$ of size 
$O(|\contextfreegrammar_1|+|\contextfreegrammar_2|)$ such that 
$\mylang(\contextfreegrammar_1\cup \contextfreegrammar_2) = \mylang(\contextfreegrammar_1)\cup \mylang(\contextfreegrammar_2)$. 
\end{observation}

The next lemma states that the symmetric grammar complexity of a group $\agroup$ is at most the index of a subgroup
$\anothergroup$ in $\agroup$ times the symmetric grammar complexity of $\anothergroup$. 
Recall that if $\agroup$ is a group and $\anothergroup$ is a subgroup of $\agroup$, 
then the index of $\anothergroup$ in $\agroup$ is 
defined as $\indexgroup_{\agroup}(\anothergroup) = \frac{|\agroup|}{|\anothergroup|}$.

\begin{lemma}
\label{lemma:ContextFreeTradeoff}
Let $\anothergroup \subgroup \agroup \subgroup \symmetricgroup_n$. 
Then $ \sgc(\agroup) \leq  \indexgroup_{\agroup}(\anothergroup) \cdot \sgc(\anothergroup)$.
\end{lemma}
\begin{proof}
Let $\cosetrepresentatives$ be a left transversal of $\anothergroup$ in $\agroup$. In other words, 
$\cosetrepresentatives$ contains exactly one element for each coset of $\anothergroup$ in $\agroup$.
Then we have that $|\cosetrepresentatives| = \indexgroup_{\agroup}(\anothergroup)$ and that 
$\agroup =  \bigcup_{\anotherpermutation\in \cosetrepresentatives} \anotherpermutation\circ \anothergroup$.
This implies that 
\begin{equation}
\label{equation:UnionCosets}
\stringfication(\agroup) =  
\bigcup_{\anotherpermutation\in\cosetrepresentatives} \stringfication(\anotherpermutation\circ \anothergroup). 
\end{equation}
From Equation \ref{equation:UnionCosets}, it is immediate that for each permutation $\alpha\in \symmetricgroup_n$,
\begin{equation}
\label{equation:PermutedUnionCosets}
\permutedcoordinates(\stringfication(\agroup),\apermutation) =  
\bigcup_{\anotherpermutation\in\cosetrepresentatives}  \permutedcoordinates(\stringfication(\anotherpermutation\circ \anothergroup),\apermutation).
\end{equation}
Now, let $\sgc(\anothergroup)$ be the symmetric context-free complexity of $\anothergroup$. 
Then, for some permutation $\apermutation\in \symmetricgroup_n$, there exists a context-free 
grammar $\contextfreegrammar$ of size $\sgc(\anothergroup)$ such that 
$\mylang(\contextfreegrammar) =  \permutedcoordinates(\stringfication(\anothergroup),\apermutation)$. 
Therefore, by Proposition \ref{proposition:RenamingProposition}, for each $\anotherpermutation\in\symmetricgroup_n$, 
there is a context-free grammar $\contextfreegrammar_{\anotherpermutation}$ of size $\sgc(\anothergroup)$ 
accepting the language $\permutedcoordinates(\stringfication(\anotherpermutation\circ \anothergroup),\apermutation)$. 
By combining Equation \ref{equation:PermutedUnionCosets} with Observation \ref{observation:UnionClosure}, 
we can infer that there is a context-free grammar $\contextfreegrammar'$ of size at most 
$\sum_{\anotherpermutation\in\cosetrepresentatives} |\contextfreegrammar_{\anotherpermutation}| = \indexgroup_{\agroup}(\anothergroup)\cdot \sgc(\anothergroup)$
	accepting the language $\permutedcoordinates(\stringfication(\agroup),\apermutation)$. 
\end{proof}

It has been shown in \cite{Ellul2004} (Theorem 30) that the language $\stringfy{\symmetricgroup_n}$ 
cannot be represented by context-free grammars of polynomial size. 
Since $\stringfy{\symmetricgroup_n}$ is invariant under permutation of coordinates, i.e., 
$\stringfy{\symmetricgroup_n} = \permutedcoordinates(\stringfy{\symmetricgroup_n},\apermutation)$ 
for any permutation $\permutation\in \symmetricgroup_n$, 
we have that the symmetric context-free complexity of $\stringfy{\symmetricgroup_n}$ is exponential.

\begin{theorem}[Theorem 30 of \cite{Ellul2004}]
\label{theorem:SymmetricGroupLowerBound}
$\sgc(\symmetricgroup_n)\geq 2^{\Omega(n)}$.
\end{theorem}

Now, by combining Theorem \ref{theorem:SymmetricGroupLowerBound} with Lemma \ref{lemma:ContextFreeTradeoff} (for $\agroup=\symmetricgroup_n$),
we have the following immediate corollary. 

\begin{corollary}
\label{corollary:IndexExponential}
Let $\anothergroup$ be a subgroup of $\symmetricgroup_n$. 
Then 
$$\sgc(\anothergroup) \geq \frac{2^{\Omega(n)}}{\indexgroup_{\symmetricgroup_n}(\anothergroup)}.$$
\end{corollary}

By combining Theorem \ref{theorem:MainTheoremEmbeddableGroup} with Corollary \ref{corollary:IndexExponential},
we have a trade-off between the size of a group $\anothergroup$, and the number of vertices, the treewidth and the 
maximum degree of a graph embedding $\anothergroup$. Below, we write $\exp_2(x)$ to denote $2^{x}$.  

\begin{theorem}
\label{theorem:TradeoffTreewidth}
There exist positive real constants $c_1,c_2$ and $c_3$ such that for large enough $n$, and 
each subgroup  $\anothergroup$ of $\symmetricgroup_n$, if $\anothergroup$ is embeddable in a graph
with $m$ vertices, maximum degree $\maximumdegree$ and treewidth $\treewidth$, then 
$$m \geq \exp_2\left( \frac{c_1 n - c_2 \treewidth \maximumdegree\log \maximumdegree - c_3\log \indexgroup_{\symmetricgroup_n}(\anothergroup)}{\treewidth}\right).$$ 
\end{theorem}
\begin{proof}
By Corollary \ref{corollary:IndexExponential}, there is a constant $c_1'$ such that for each sufficiently large $n$, 
	each subgroup $\anothergroup$ of $\symmetricgroup_n$, $$\sgc(\anothergroup)\geq  \frac{2^{c_1' n}}{\indexgroup_{\symmetricgroup_n}(\anothergroup)}.$$ 
By Theorem \ref{theorem:MainTheoremEmbeddableGroup}, there exists constants $c_2'$ and $c_3'$ such that for each large enough $n$, 
if a subgroup $\anothergroup$ of $\symmetricgroup_n$ can be embedded in a graph with $m\geq n$ vertices, maximum degree $\maximumdegree$ and treewidth $\treewidth$
then $$\sgc(\anothergroup) \leq 2^{c_2' \treewidth \maximumdegree\log\maximumdegree}\cdot m^{c_3' \treewidth}.$$ Therefore, we have that for sufficiently 
large $n$, if a subgroup of $\symmetricgroup_n$ can be embedded in a graph with $m$ vertices, maximum degree $\maximumdegree$ and treewidth $\treewidth$, then 
$$m \geq \left( \frac{2^{c_1' n}}{\indexgroup_{\symmetricgroup_n}(\anothergroup)\cdot 2^{c_2' k \maximumdegree\log \maximumdegree}} \right)^{1/c_3'k}.$$
	Therefore, by setting $c_1 = c_1'/c_3'$, $c_2 = c_2'/c_3'$ and $c_3= 1/c_3'$, we have that 
$$m \geq \exp_2\left( \frac{c_1 n - c_2 \treewidth \maximumdegree\log \maximumdegree - c_3\log \indexgroup_{\symmetricgroup_n}(\anothergroup)}{\treewidth}\right).$$ 
\end{proof}

As a corollary of Theorem \ref{theorem:TradeoffTreewidth}, we get the following lower bound stating that 
subgroups of $\symmetricgroup_n$ with small index (i.e. index at most $2^{cn}$ for some small constant $c$) cannot be 
embedded on graphs of treewidth $o(n/\log n)$, maximum degree $o(n/\log n)$ and a polynomial number of vertices. 

\begin{corollary}
\label{corollary:LowerBoundOne}
Let $\graphclass$ be a class of connected graphs of treewidth $o(n/\log n)$ and maximum-degree $o(n/\log n)$. Then there is a function
	$f\in \omega(1)$, and a constant $c\in \mathbb{R}$, such that for each sufficiently large $n$, each
	subgroup $\agroup$ of $\symmetricgroup_n$ of index $\indexgroup_{\symmetricgroup_n}(\agroup)\leq 2^{cn}$
has $\graphclass$-embedding complexity at least $n^{f(n)}$. 
\end{corollary}

For classes of graphs of treewidth $n^{\varepsilon}$ (for $\varepsilon<1$), and maximum degree $o(n/\log n)$, Theorem 
\ref{theorem:MainTheoremEmbeddableGroup} implies exponential lower bounds on the embedding complexity of groups 
of small index (i.e. index at most $2^{cn}$ for some small constant $c$). 

\begin{corollary}
\label{corollary:LowerBoundTwo}
	Let $\graphclass$ be a class of connected graphs of treewidth $n^{\varepsilon}$ (for $\varepsilon < 1$) and maximum-degree $o(n/\log n)$. 
	Then there exist constants $c,c'\in \mathbb{R}$, such that for each sufficiently large $n$, each subgroup $\agroup$ of 
	$\symmetricgroup_n$ of index $\indexgroup_{\symmetricgroup_n}(\agroup)\leq 2^{cn}$ 
	has $\graphclass$-embedding complexity at least $2^{c'n^{1-\varepsilon}}$. 
\end{corollary}

In particular, for some small $c,c'\in \R$, the graph embedding complexity of 
subgroups of $\symmetricgroup_n$ of index at most $2^{c'n}$ is lower bounded by 
$2^{c\sqrt{n}}$ for any minor closed class of graphs of maximum degree $o(n\log n)$. Note that 
these classes of graphs have treewidth at most $\sqrt{n}$.

\section{Conclusion and Open Problems}
\label{section:OpenProblems}

In this work, we have established new connections between three complexity measures for permutation groups: 
embedding complexity parameterized by treewidth and maximum-degree, symmetric grammar complexity and 
extension complexity. In particular, we have shown that groups that can be embedded in graphs of small treewidth
and degree have small symmetric grammar complexity and small extension complexity. These results can also be 
used to translate strong lower bounds on the symmetric grammar complexity or on the extension complexity of 
a group $\agroup \subgroup \symmetricgroup_n$ into lower bounds on the embedding complexity of $\agroup$. 
In particular, using this approach, we have shown that subgroups $\agroup \subgroup \symmetricgroup_n$ of 
sufficiently small index have superpolynomial embedding complexity on classes of graphs of treewidth $o(n/\log n)$ and 
maximum degree $o(n/\log n)$.

Below, we state some interesting open problems related to our work. 

\begin{problem}
\label{problem:SomeGroupSuperpolynomial}
Construct an explicit family of groups $\{\agroup_n\}_{n\in \Nplus}$ with superpolynomial graph embedding complexity, that is to 
say, such that $\embeddingcomplexity(\agroup_n) = n^{\Omega(1)}$.
\end{problem}

In particular, it is not known if the graph embedding complexity of the alternating group $\alternatinggroup_n$ is superpolynomial.
Note that the graph embdding complexity of the symmetric group $\symmetricgroup_n$ is $n$, which is witnessed by $K_n$, the complete 
graph with vertex set $\{1,\dots,n\}$. 

\begin{problem}
\label{problem:AlternatingGroupSuperpolynomial}
Does the alternating group $\alternatinggroup_n$ have superpolynomial graph embedding complexity? 
\end{problem}

The $n$-alternahedron polytope $\apolytope(\alternatinggroup_n)$ is the polytope associated with the alternating group $\alternatinggroup_n$.
The technique used in \cite{Goemans2015smallest} to prove an $O(n\log n)$ upper bound on the extension complexity of the $n$-permutahedron 
$\apolytope(\symmetricgroup_n)$ was generalized in \cite{Weltge2012erweiterte} to show that the extension complexity of the 
$n$-alternahedron is $O(n\log n)$. Therefore, if the answer to Problem \ref{problem:AlternatingGroupSuperpolynomial} is positive, the 
alternating group $\alternatinggroup_n$ is also a solution to the following problem. 

\begin{problem}
\label{problem:PolynomialECSuperpolynomialGEC}
Construct a family of groups $\{\agroup_n\}_{n\in \Nplus}$ of polynomial extension complexity and superpolynomial graph embedding complexity. 
\end{problem}

It is also worth noting that no superpolynomial lower bound for the extension complexity of permutation groups has been shown yet. 

\begin{problem}
	\label{problem:ExtensionComplexityLowerBoundGroups}
	Construct an explicit family of groups $\{\agroup_n\}_{n\in \Nplus}$ with superpolynomial extension complexity. 
\end{problem}

We note that the upper bound of $O(n\log n)$ proved in \cite{Goemans2015smallest} and in \cite{Weltge2012erweiterte} on the extension complexity 
of $\symmetricgroup_n$ and $\alternatinggroup_n$ repectively are with respect the representation of permutations as strings of length $n$ 
over the alphabet $[n]$. In the realm of linear programming theory, another 
useful way of representing permutations of the set ${1,\dots,n}$ is as $0/1$-permutation matrices of
dimension $n$. In this case, the polytope associated with a permutation group $\agroup\subgroup \symmetricgroup_n$ is the polytope $\apolytopematrices(\agroup)$ formed 
by the convex hull of permutation matrices corresponding to the elements of $\agroup$. 
In the case of the symmetric group $\symmetricgroup_n$, the polytope $\apolytopematrices(\symmetricgroup_n)$ is 
the well known Birkhoff polytope \cite{Birkhoff1946} that has extension complexity $\Theta(n^2)$.
On the other hand, determining whether the polytope $\apolytopematrices(\alternatinggroup_n)$
has polynomial extension complexity is an important open problem in linear-programming theory. 

\begin{problem}
\label{problem:AlternatingMatrices}
Does the polytope $\apolytopematrices(\alternatinggroup_n)$ have polynomial extension complexity? 
\end{problem}

\medskip
\noindent{\bf Acknowledgements.}
We thank Manuel Aprile, Laszlo Babai, Peter Cameron, Michael Fellows and Samuel Fiorini for valuable comments and suggestions.
We thank Michel Goemans, Kanstantsin Pashkovich and Stefan Weltge for answering some of our questions by email.

\bibliographystyle{abbrv}
\bibliography{references}

\end{document}